\documentclass[times,sort&compress,3p]{elsarticle} 
\journal{Journal of Multivariate Analysis}
\usepackage{geometry,graphicx,color}
\usepackage{amsmath, amssymb, amsfonts, amsthm, float,pbox, multirow}
\usepackage{comment, longtable, caption, subcaption}

\allowdisplaybreaks

\newcommand{\pcite}[1]{\citeauthor{#1}'s \cite{#1}}

\newtheorem{theorem}{Theorem}

\newtheorem{proposition}{Proposition}
\newtheorem{lemma}{Lemma}
\newtheorem{corollary}{Corollary}
\theoremstyle{remark}

\newtheorem{remark}{Remark}

\newcommand{\cX}{\mathcal{X}}
\newcommand{\real}{\mathbb{R}}

\begin{document}
\begin{frontmatter}

\title{Multivariate initial sequence estimators in Markov chain Monte Carlo}

\author{Ning Dai and Galin L. Jones}
\address{School of Statistics, University of Minnesota, daixx224@umn.edu, galin@umn.edu}

\begin{abstract}
  Markov chain Monte Carlo (MCMC) is a simulation method commonly used
  for estimating expectations with respect to a given distribution. We
  consider estimating the covariance matrix of the asymptotic
  multivariate normal distribution of a vector of sample
  means. \citet{geye:1992} developed a Monte Carlo error estimation
  method for estimating a univariate mean. We propose a novel
  multivariate version of Geyer's method that provides an
  asymptotically valid estimator for the covariance matrix and results
  in stable Monte Carlo estimates. The finite sample properties of the
  proposed method are investigated via simulation experiments.
\end{abstract}


\end{frontmatter}

\section{Introduction}
\label{sec:introduction}

Many distributions encountered in modern applications are intractable
in the sense that it is difficult to calculate expectations without
resorting to simulation-based methods.  If it is difficult to simulate
independent realizations from the target distribution, then it is
natural to turn to Markov chain Monte Carlo (MCMC).  An MCMC
experiment consists of generating a realization of an irreducible
Markov chain having the distribution of interest as its stationary
distribution \citep{liu:2001, robe:case:1999}. The simulated data may
then be used to estimate a vector of means associated with the
stationary distribution. The reliability of this estimation can be
assessed by forming asymptotically valid confidence regions for the
means of the stationary distribution \citep{fleg:jone:2011,
  fleg:hara:jone:2008, geye:1992, jone:hobe:2001,
  jone:hara:caff:neat:2006,vats:fleg:jone:moa:2015}. (There is a
simliar approach to quantile estimation
\cite{doss:fleg:jone:neat:2014}.)  The confidence regions are based on
estimating the covariance matrix in a multivariate Markov chain
central limit theorem (CLT).  We propose and study a novel method for
estimating this covariance matrix.

Estimating the covariance matrix has been mostly ignored in the MCMC
literature until recently. \citet{vats:fleg:jone:moa:2015} and \citet{vats:fleg:jone:sve:2016} studied non-overlapping batch means and
spectral methods, respectively, and found that these estimators often
underestimate the size of the confidence regions and overestimate the
effective sample size unless the Monte Carlo sample sizes are
enormous. \citet{Kosorok2000} proposed an estimator that is closer in
spirit to ours than the spectral and batch means methods, but we will
see later that it typically overestimates the effective sample size,
resulting in overconfidence in the reliability of the simulation.  We propose
alternative estimators of the covariance matrix that require weaker
mixing conditions on the Markov chain and weaker moment conditions on
the function of interest than those required by batch means and
spectral methods. Specifically, our method applies as long as a Markov
chain CLT holds and detailed balance is satisfied, which is not enough
to guarantee the asymptotic validity of batch means or spectral
methods.  We show that the proposed estimators are asymptotically
valid and study their empirical performance.  The problem we consider
will now be described more formally.

Let $F$ be a distribution having support $\cX$ and if $p \ge 1$, let
$g : \cX \to \real^p$ be $F$-integrable and set
\[
\mu = \mathrm{E}_{F} \left\{g(X)\right\}=\int_{\cX}g(x)F(dx).
\]
Also let $\Phi=\{X_0,X_1,X_2,\ldots\}$ be a Harris ergodic---namely, irreducible,
aperiodic and Harris recurrent---Markov chain having invariant
distribution $F$. By averaging the function over a realization of
$\Phi$, estimation of $\mu$ is straightforward since, with probability
1,
\[
\mu_n=\frac{1}{n}\sum\limits_{i=1}^n g(X_i) \to \mu~~~ \text{ as } ~n \to \infty.
\]
The Markov chain strong law justifies the use of MCMC but provides no
information about the quality of estimation or how large the
simulation size $n$ should be.  More specifically, additional
information is needed to answer either of the following two questions.
\begin{enumerate}
\item Given a pre-specified run length $n$, how reliable is $\mu_n$ as
  an estimate of $\mu$? Specifically, how do we construct a confidence
  region for $\mu$? 
\item How large should the simulation size $n$ be to ensure a
  reliable estimate of $\mu$?
\end{enumerate}
We can address these issues through the approximate sampling
distribution of the \textit{Monte Carlo error}, $\mu_n-\mu$.  A Markov
chain CLT exists when there is a positive definite matrix $\Sigma$
such that, as $n\to\infty$,
\begin{equation}\label{eq:multiCLT}
\sqrt{n} \, (\mu_n-\mu)\rightsquigarrow\mathcal{N}_{p} (0,\Sigma).
\end{equation}
See \citet{jone:2004} and \citet{robe:rose:2004} for conditions which
ensure a CLT.  Notice that, due to the serial correlation inherent to
the Markov chain, $\Sigma\neq \mathrm{var}_{F} \{g(X)\}$ except in
trivial cases. In Section~\ref{sec:construction} we propose two new
estimators of $\Sigma$. For now, let $\Sigma_n$ be a generic positive
definite estimator of $\Sigma$.

A confidence region for $\mu$ constructed using $\Sigma_n$ forms an
ellipsoid in $p$ dimensions oriented along the directions of the
eigenvectors of $\Sigma_n$. Let $|\cdot|$ denote determinant. One can
verify by straightforward calculation that the volume of the
confidence region is proportional to $\sqrt{|\Sigma_n|}$ and thus
depends on the estimated covariance matrix $\Sigma_n$ only through the
estimate $|\Sigma_n|$ of the \textit{generalized variance of the Monte
  Carlo error}, $|\Sigma|$. The volume of the confidence region can
describe whether the simulation effort is sufficiently large to
achieve the desired level of precision in estimation
\citep{jone:hara:caff:neat:2006, fleg:hara:jone:2008,
  vats:fleg:jone:moa:2015}.

Another common and intuitively reasonable method for choosing the
simulation effort is to simulate until a desired effective sample size
(ESS), i.e., the number with the property that $\mu_n$ has the same
precision as the sample mean obtained by that number of independent
and identically distributed (iid) samples, has been achieved
\citep{Atkinson2008, Drummond2006, Giordano2015}. Let
$\Lambda= \mathrm{var}_{F}\{g(X)\}$. \citet{vats:fleg:jone:moa:2015}
introduced the following definition of effective sample size
\begin{equation}
\label{eq:ESS}
\text{ESS} = n\left({|\Lambda|}/{|\Sigma|}\right)^{{1}/{p}},
\end{equation}
which is naturally estimated with $n(|\Lambda_n|/|\Sigma_n|)^{1/p}$
where $\Lambda_n$ is an estimator of $\Lambda$, e.g., the usual sample
covariance matrix. \citet{vats:fleg:jone:moa:2015} showed that
terminating the simulation based on the effective sample size is
equivalent to termination based on a relative confidence region where
the Monte Carlo error is compared to size of the uncertainty in the
target distribution.  The point is that again a common method for
assessing the reliability of the simulation is determined by the
estimated generalized variance of the Monte Carlo error.

The estimators of $\Sigma$ studied by \citet{Kosorok2000}, \citet{vats:fleg:jone:moa:2015}, and \citet{vats:fleg:jone:sve:2016}
typically underestimate the generalized variance. We will propose 
a different method and show that it is asymptotically valid. Specifically, our method provides a consistent overestimate for the asymptotic generalized variance of the Monte Carlo error and therefore will result in a slightly larger simulation effort, leading to a more stable estimation process.

The rest of the paper is organized as follows. In
Section~\ref{sec:notation} we develop notation and background in
preparation for the estimation theory. In
Section~\ref{sec:construction} we propose our method and establish its
asymptotic validity. In Section~\ref{sec:simulation} we examine the
finite sample properties of the proposed method through a variety of
examples. We consider a Bayesian logistic regression example of 5
covariates where a symmetric random walk Metropolis--Hastings algorithm
is implemented to calculate the posterior mean of the regression
coefficient vector, a Bayesian one-way random effects model where we
use a random scan Gibbs sampler to estimate the posterior expectation
of all 8 parameters, and a reversible multivariate AR(1) process that
takes values in $\real^{12}$. We illustrate the use of multivariate
methods in a meta-analysis application where the posterior has
dimension 65.

\section{Notation and background}
\label{sec:notation}

Recall that $F$ has support $\cX$ and let $\mathcal{B}(\cX)$ be a
$\sigma$-algebra. For $n \in \mathbb{N}^+ = \{1, 2, 3,\ldots\}$ let
$P^{n}(x,dy)$ be the $n$-step Markov transition kernel so that for $x
\in \cX$, $B \in \mathcal{B}(\cX)$, and $k \in \mathbb{N}=\{0,1,2,\ldots\}$ we
have $P^{n}(x,B) = \Pr (X_{k+n} \in B \mid X_{k} = x)$, where $\Pr$ denotes probability.  We assume that
$P$ satisfies detailed balance with respect to $F$.  That is,
\begin{equation}
\label{eq:dbc}
F(dx) P(x,dy) = F(dy) P(y, dx) \; .
\end{equation}
Metropolis--Hastings algorithms satisfy \eqref{eq:dbc} by construction
as do many component-wise Markov chains, such as random scan or random
sequence scan algorithms \cite{john:jone:neat:2013}.  By integrating
both sides of \eqref{eq:dbc} it is easy to see that $F$ is invariant
for $P$.  Suppose $X_0 \sim F$, that is the Markov chain is
stationary.  The assumption of stationarity is not crucial since, for
Harris recurrent chains, if a CLT holds under stationarity, it holds
for all initial distributions \cite[][Proposition~17.1.6]{MeynTweedie1993}.

The lag $t$ autocovariance of the process $g(X_0), g(X_1), g(X_2),
\ldots$ is defined as
$
\gamma_t=\gamma_{-t}=\mathrm{cov}_F\{ g(X_i),g(X_{i+t})\}.
$
Denote the sum of an adjacent pair of
autocovariances by $\Gamma_i=\gamma_{2i}+\gamma_{2i+1}$ for $i\in\mathbb{N}$ and its smallest eigenvalue by $\xi_i$. 

We use the shorthand $\infty$ for $+\infty$ unless otherwise specified. If $\sum_{t=0}^\infty\gamma_t$ converges, the asymptotic covariance matrix in \eqref{eq:multiCLT} can be written as \cite{KipnisVaradhan1986}
\begin{equation}\label{eq:Sigma_inf}
\Sigma = \sum_{t=-\infty}^{+\infty} \gamma_t=-\gamma_0 +
\sum_{t=0}^{\infty}( \gamma_t + \gamma_{-t}) = -\gamma_0 + 2
\sum_{t=0}^{\infty} \gamma_t =-\gamma_0+2
\sum_{i=0}^{\infty}\Gamma_i.
\end{equation}

The following propositions will play a significant role in the
development of the new estimation method in
Section~\ref{sec:construction}.

\begin{proposition}\label{prop:1}
The following properties of the sequences $\{\Gamma_i: i\in \mathbb{N}\}$
and $\{\xi_i : i\in \mathbb{N}\}$ hold.
\begin{enumerate}
\item [(i)] $\Gamma_i$ is positive-definite, for all $i\in\mathbb{N}$.\label{prop1:1}
\item [(ii)] $\Gamma_i-\Gamma_{i+1}$ is positive-definite, for all $i\in\mathbb{N}$.\label{prop1:2}
\item [(iii)] $\lim_{i\to\infty}\Gamma_i=0$.\label{prop1:3}
\item [(iv)] The sequence $\{\xi_i : i\in \mathbb{N}\}$ is positive, decreasing, and converges to $0$.\label{prop1:4}
\end{enumerate}
\end{proposition}

\begin{proof}
See Appendix~\ref{app:propositions}.
\end{proof}

Recall \eqref{eq:Sigma_inf} and let the $m$th partial sum be denoted
\begin{equation}\label{eq:Sigma_m}
\Sigma_m =-\gamma_0+\sum\limits_{t=0}^{2m+1}(\gamma_t+\gamma_{-t})=-\gamma_0 + 2\sum_{i=0}^m \Gamma_i  .
\end{equation}

\begin{proposition}\label{prop:2}
The following properties of the sequence $\{\Sigma_m : m\in \mathbb{N}\}$ hold.
\begin{enumerate}
\item [(i)] There exists a non-negative integer $m_0$ such that $\Sigma_m$
  is positive definite for $m\geq m_0$ and not positive definite for
  $m< m_0$. Specifically, when $m_0=0$, $\Sigma_m$ is positive
  definite for all $m$.\label{prop2:1}
\item [(ii)] The sequence $\{|\Sigma_m| :  m=m_0,m_0+1,m_0+2,\ldots\}$ is positive,
  increasing, and converges to $|\Sigma|$.\label{prop2:2}
\end{enumerate}
\end{proposition}

\begin{proof}
See Appendix~\ref{app:propositions}.
\end{proof}

\begin{remark}\label{rm:asp}
  The value of $m_0$ is difficult to calculate explicitly because
  $\Sigma_m$ is usually not available in closed form. However, in
  Section~\ref{sub:AR1} we consider a multivariate AR(1) Markov chain
  and verify that $m_0=0$. In the other examples, we cannot establish
  $m_0=0$ directly, but in our simulations we never observed anything
  else in 2000 independent replications.
  \end{remark}

\section{Estimation method}
\label{sec:construction}
A natural estimator of the lagged autocovariance $\gamma_t$ is the
empirical autocovariance
\begin{equation*}
\gamma_{n,t} = \gamma_{n,-t}^\top = \frac{1}{n} \sum_{i=1}^{n-t}
\left\{g(X_i)-\mu_n\right\} \left\{g(X_{i+t})-\mu_n\right\}^\top \,
\end{equation*}
where $^\top$ denotes transpose. Set
$
\widetilde{\gamma}_{n,t}= (\gamma_{n,t}+\gamma_{n,-t})/2
$
for $t \in \{0,\ldots,n-1\}$ and write the sum of the $i$th ($0\leq i\leq \lfloor n/2-1\rfloor$) adjacent pair as
$
\Gamma_{n,i}=\widetilde{\gamma}_{n,2i}+\widetilde{\gamma}_{n,2i+1}.
$
By construction, $\Gamma_{n,i}$ is symmetric. Let $\xi_{n,i}$ denote its smallest eigenvalue. The empirical estimator of $\Sigma_m$ ($0\leq m\leq \lfloor n/2-1\rfloor$) is
\begin{equation}\label{eq:Sigma_n,m}
\Sigma_{n,m}
=-\gamma_{n,0}+\sum\limits_{t=0}^{2m+1}(\gamma_{n,t}+\gamma_{n,-t})
=-\gamma_{n,0}+2\sum\limits_{i=0}^m\Gamma_{n,i}.
\end{equation}
Notice how \eqref{eq:Sigma_n,m} parallels \eqref{eq:Sigma_m}. 

\subsection{Multivariate initial sequence estimators}
\label{sub:extending}

We are now in position to formally define the multivariate initial
sequence (mIS) estimator.  Let $s_n$ be the smallest integer such that
$\Sigma_{n,s_n}$ is positive definite and let $t_n$ be the largest
integer $m$ ($s_n\leq m\leq \lfloor n/2-1\rfloor$) such that
$|\Sigma_{n,i}|>|\Sigma_{n,i-1}|$ for all $i\in \{s_n+1,\ldots,m\}$. Then the mIS estimator, denoted $\Sigma_{\mathrm{seq},n}$, is defined as
$
\Sigma_{\mathrm{seq},n}=\Sigma_{n,t_n} .
$
It is possible that $\Sigma_{n,m}$ fails to be positive definite for all
$m \in \{0,\ldots, \lfloor n/2-1\rfloor\}$, and consequently $s_n$
does not exist. Fortunately, when $n$ is sufficiently large, we can
always find such $s_n$.

\begin{theorem}\label{th:feasible}
  With probability 1, $s_n$ exists as $n\to\infty$.  In particular,
  with probability 1, $s_n \to m_0$ as $n\to\infty$.
\end{theorem}
\begin{proof}
See Appendix \ref{app:theorems}.
\end{proof}

Thus mIS is feasible while the following establishes that it is asymptotically valid.

\begin{theorem}\label{th:mIS}
With probability 1, $\liminf\limits_{n\to\infty}
|\Sigma_{\mathrm{seq},n}| \geq |\Sigma|$.
\end{theorem}
\begin{proof}
See Appendix \ref{app:theorems}.
\end{proof}

In the construction of $\Sigma_{\mathrm{seq},n}$ we update
$\Sigma_{n,i}$ to $\Sigma_{n,i+1} = \Sigma_{n,i} +
2\Gamma_{n,i+1}$. If $\Gamma_{n,i+1}$ has negative eigenvalues, adding
$2\Gamma_{n,i+1}$ will squeeze the corresponding confidence region in
undesirable directions. A remedy is to force the negative eigenvalues
of $\Gamma_{n,i+1}$ to be 0. Suppose $\Gamma_{n,i+1}$ has eigen-decomposition $\Gamma_{n,i+1} = Q^\top \Lambda Q$ where $\Lambda = \mathrm{diag} 
(\lambda_1,\ldots,\lambda_p)$. Define the positive part of $\Gamma_{n,i+1}$ as $\Gamma_{n,i+1}^+ = Q^\top\Lambda^+ Q$, where $\Lambda^+ =
  \mathrm{diag}(\max\{\lambda_1,0\}, \ldots,
  \max\{\lambda_p,0\})$.

This leads us to define the adjusted multivariate initial sequence (mISadj) estimator.
Let $s_n$ and $t_n$ be as in the definition of mIS and let
\[
\widetilde{\Sigma}_{n,t_n}=\Sigma_{n,s_n}+2\sum_{i=s_n+1}^{t_n}\Gamma_{n,i}^+
\]
where $\Gamma^+_{n,i}$ is the positive part of $\Gamma_{n,i}$. Then the
mISadj estimator, denoted $\Sigma_{\mathrm{adj},n}$, is defined as
$
\Sigma_{\mathrm{adj},n}=\widetilde{\Sigma}_{n,t_n}.
$ See Figure~\ref{fig:adjdemo} for a display of the effect of using mISadj over mIS.
\begin{figure}[b!]
  \centering
    \includegraphics[width=0.35\linewidth]{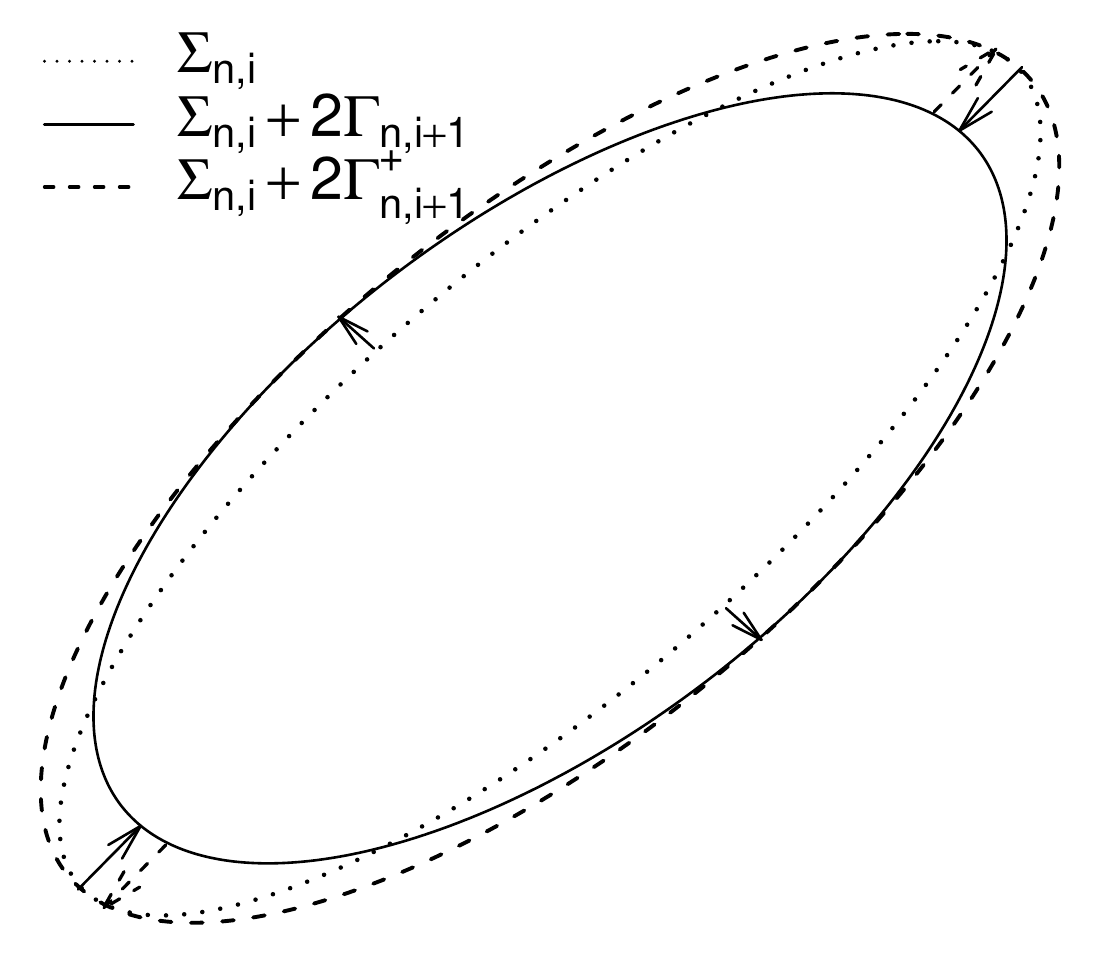}
    \caption{A diagrammatic sketch of the adjustment effect on a
      confidence region. Adding $2\Gamma_{n,i+1}$ squeezes the
      confidence region in the direction of the eigenvector corresponding to the negative eigenvalue. The adjustment cancels the
      shrinkage.}
    \label{fig:adjdemo}
\end{figure}

By construction, the mISadj estimator is positive definite. The modification adds a positive semi-definite matrix to the mIS estimator, which by Theorem \ref{th:mIS} provides a consistent overestimate for the generalized variance, $|\Sigma|$, and therefore the mISadj estimator also has a larger determinant than the asymptotic covariance matrix, $\Sigma$.
\begin{theorem}\label{th:mISadj}
  With probability 1, $\liminf\limits_{n\to\infty}
  |\Sigma_{\mathrm{adj},n}| \geq |\Sigma|$.
\end{theorem}

\subsubsection{Related estimators}

The motivation for our approach can be found in \pcite{geye:1992}
univariate initial positive sequence (uIS) estimator.  Suppose $\mu$
is one-dimensional and denote the variance of the asymptotic normal
distribution $\sigma^2$.  In this setting \citet{geye:1992} proposed
the uIS estimator
\[
\sigma^2_{\mathrm{pos},n} = -\gamma_{n,0} + 2 \sum\limits_{i=0}^{t_n}
\Gamma_{n,i},
\]
where $t_n$ is the largest integer $m$ such that $\Gamma_{n,i}>0$ for
all $i \in \{1,\ldots,m\}$. That is, Geyer's truncation rule is to stop adding
in $2\Gamma_{n,i}$ when it causes $\sigma^2_{n,i}=-\gamma_{n,0} + 2(\Gamma_{n,0}+ \cdots
+ \Gamma_{n,i})$ to
decrease. (Figure~\ref{fig:uniseri} depicts the behavior of
$\sigma^2_{n,i}$ and $\Gamma_{n,i}$ for one of the examples we
consider later.) The uIS estimator is therefore the
first local maximum of the sequence
$\{\sigma^2_{n,i}: i=0,\ldots,\lfloor n/2-1\rfloor\}$ and thus gives an asymptotic
overestimate of $\sigma^2$. This is formally stated in his Theorem~3.2:
\[
\liminf_{n\to\infty}\sigma^2_{\mathrm{pos},n}\geq \sigma^2
~~~\text{with probability}~1.
\]

Neither mIS nor mISadj is a straightforward generalization of Geyer's
method in that mIS and mISadj coincide but do not
reduce to uIS when $\mu$ is one-dimensional. However, this is not
essential because the three methods are asymptotically equivalent in univariate settings.

\begin{figure}[t!]
    \centering
    \includegraphics[width=0.5\linewidth]{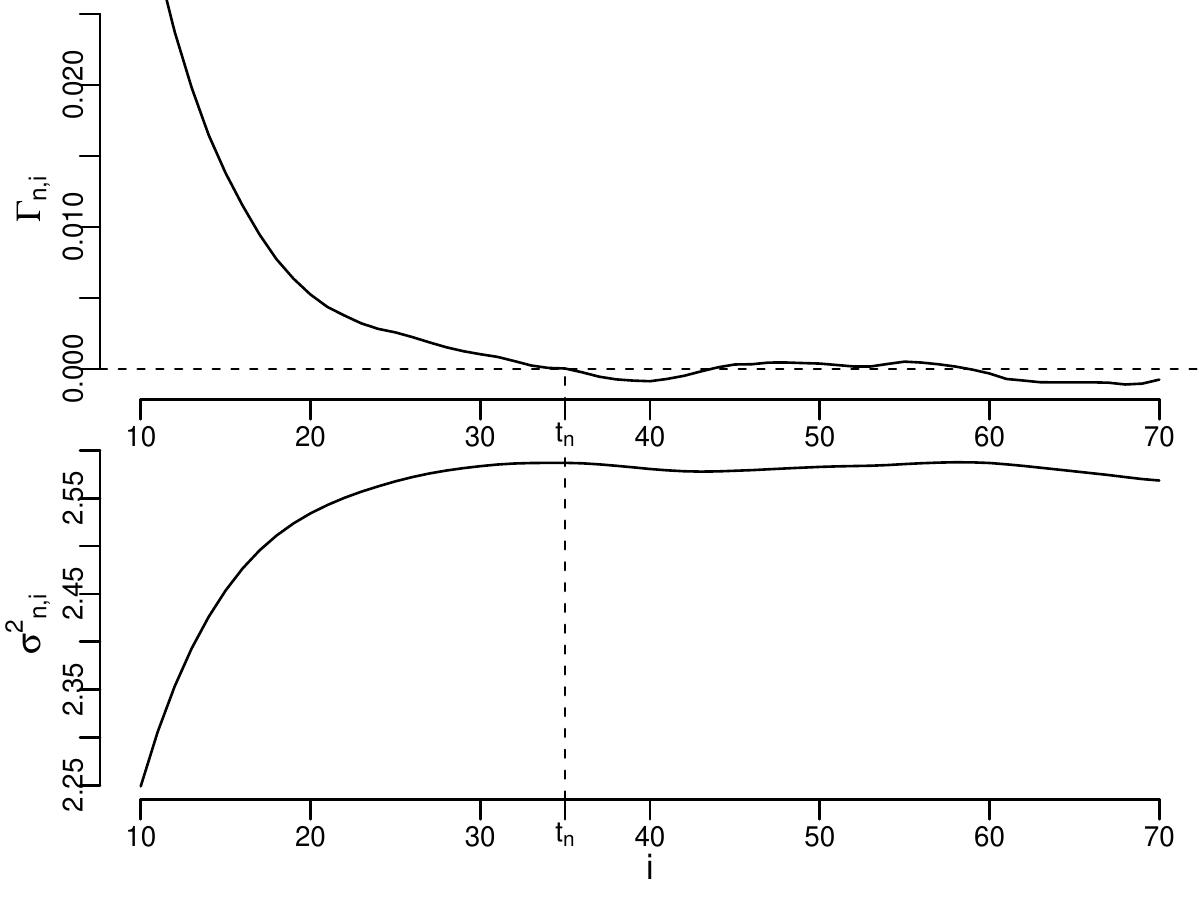}
    \caption{The uIS method truncates the first time $\Gamma_{n,i}$ is non-positive, or equivalently at the first local maximum of $\{\sigma^2_{n,i};i=0,\ldots,\lfloor n/2-1\rfloor\}$. Computed using a marginal chain of the Bayesian logistic regression example described in Section~\ref{sub:eg1} with Monte Carlo sample size $10^6$.}
    \label{fig:uniseri}
\end{figure}

\citet{Kosorok2000} proposed an alternative multivariate estimator (mK)
which was also motivated by \pcite{geye:1992} approach. Recall from
Proposition~\ref{prop:1} that $\{ \xi_i
;i \in \mathbb{N}\}$ is positive, decreasing, and converges to $0$, where $\xi_i$ is the smallest eigenvalue of $\Gamma_i$.  In mK the truncation point is chosen to be the
largest integer $m$ such that $\xi_{n,i} > 0$ for all $i \in \{1, \ldots,
m\}$. However, this does not ensure that the generalized variance is
adequately estimated and often truncates before the sequence
$\{|\Sigma_{n,i}| : i=s_n, \ldots,\lfloor n/2-1\rfloor\}$ reaches the first local
maximum, as demonstrated in Figure~\ref{fig:multiseriK} and \ref{fig:together}.

\begin{figure}[ht]
    \centering
    \includegraphics[width=0.5\linewidth]{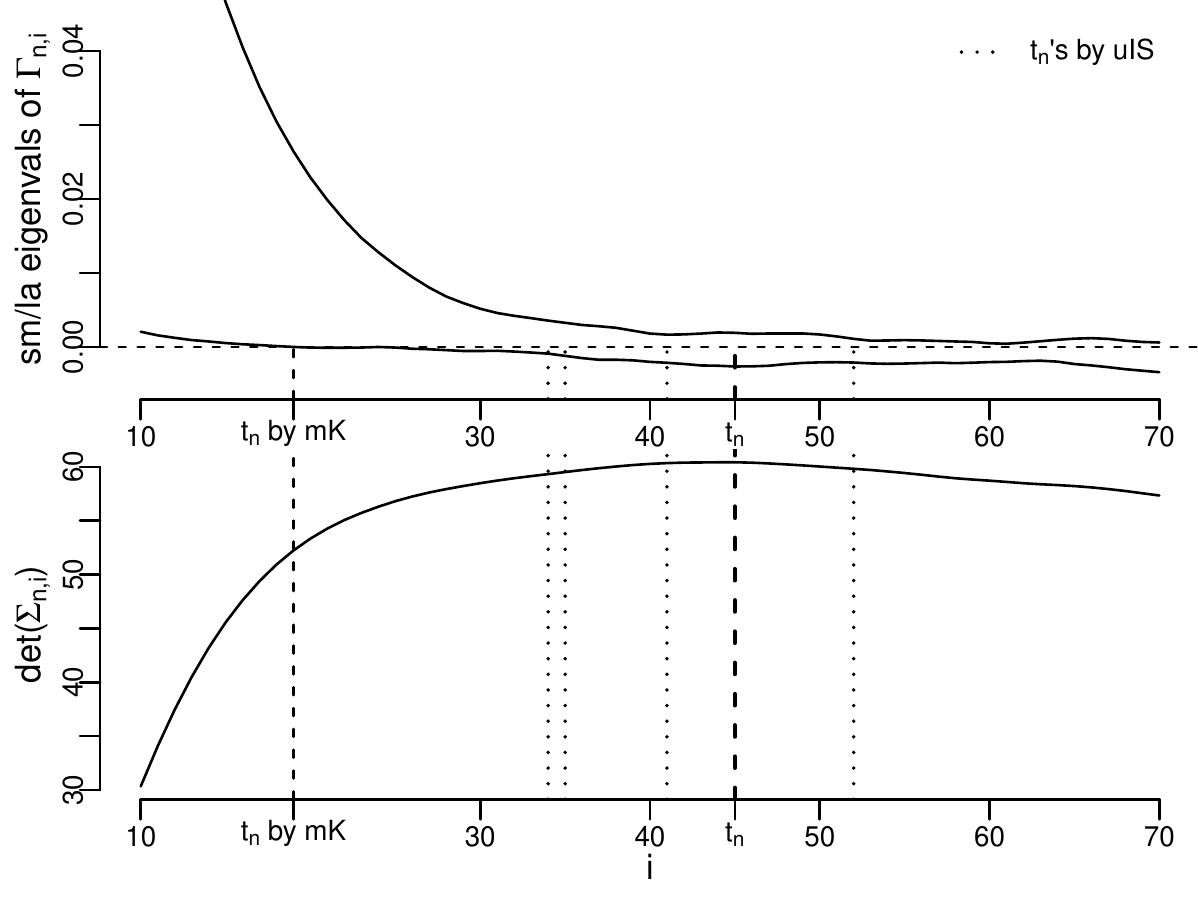}
    \caption{The thick dashed line marks the mIS truncation point. Truncating at the first local maximum of $\{|\Sigma_{n,i}|;i=s_n,\ldots,\lfloor n/2-1\rfloor\}$ achieves a balance between the individual components. The mK truncation point---namely the first time $\{\Gamma_{n,i};i=0,\ldots,\lfloor n/2-1\rfloor\}$ fails to be positive definite---is premature. Computed using the 5-dimensional Bayesian logistic regression example described in Section~\ref{sub:eg1} with Monte Carlo sample size $10^6$.}
    \label{fig:multiseriK}
\end{figure}

\begin{figure}[b!]
    \centering
    \includegraphics[width=0.5\linewidth]{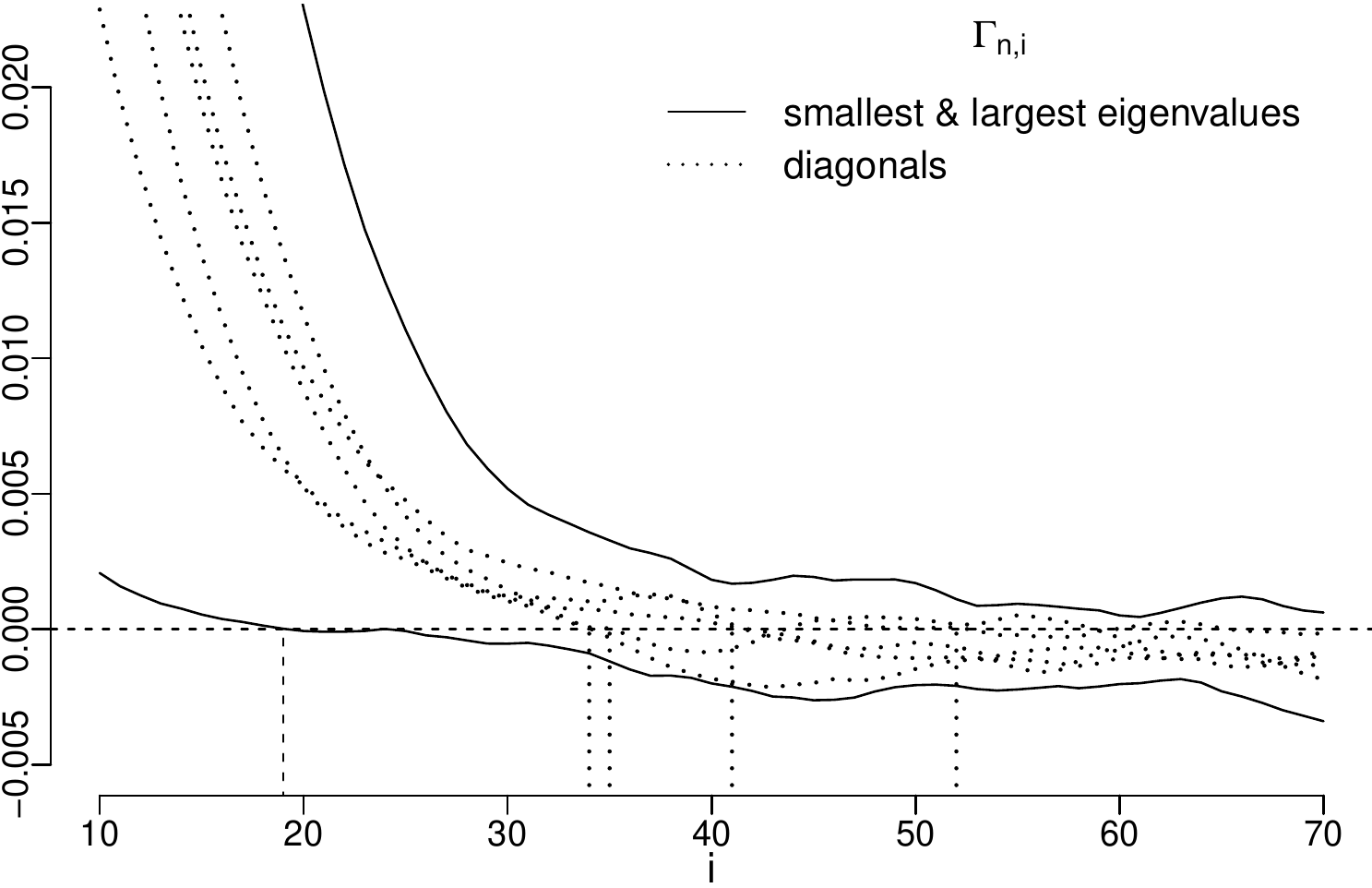}
    \caption{The diagonal entries of $\Gamma_{n,i}$, which correspond to the individual components, are always between the smallest and largest eigenvalues. It is too early to truncate the first time the smallest eigenvalue drops below 0, as is the case with mK (see the vertical dashed line). On the other extreme, it is too late to truncate when the largest eigenvalue drops below 0. The ideal truncation point should be somewhere between the uIS truncation points, marked by the vertical dotted lines. Computed using the 5-dimensional Bayesian logistic regression example described in Section~\ref{sub:eg1} with Monte Carlo sample size $10^6$.}
    \label{fig:together}
\end{figure}

\section{Simulation experiments}
\label{sec:simulation}

Our goal is to investigate the finite-sample properties of mIS,
mISadj, mK, and uIS through simulation experiments in a variety of examples. In
each of the examples, which are described in more detail below, we
compare the approaches in terms of effective sample size as well as
volume and coverage probability of a joint confidence region.

We describe the simulation examples and the MCMC algorithms used in
Section~\ref{sub:eg1}--\ref{sub:AR1}.  The results of the simulation
experiments are given in Section~\ref{sub:results}.  We then consider
a meta-analysis application in Section~\ref{sub:real}.

\subsection{Bayesian logistic regression}
\label{sub:eg1}

For $i \in \{1,\ldots, 100\}$, let $X_{i} = (x_{i1},\ldots, x_{i5})$ be the
observed covariates for the $i$th observation and $Y_i$ be the binary
response. We suppose
\[
Y_i|X_i,\beta\overset{ind}{\sim} \mathrm{Bernoulli}\left\{ \frac{1}{1 +
    \exp(-X_i\beta)} \right\}~~\text{and}~~
\beta\sim \mathcal{N}_5(0,4I_5) .
\]
This model results in a posterior on $\mathbb{R}^5$, denoted $F$.  The
data we use is provided in the \verb+logit+ dataset in the \verb+mcmc+
R package.

We are interested in estimating the posterior mean of $\beta$, i.e.,
$\mu =\mathrm{E}_{F} (\beta)$.  However, this expectation is intractable and
hence we will use a symmetric random walk Metropolis--Hastings
algorithm to estimate it. At each step of the
Markov chain, the proposal for the next step is $\mathcal{N}_5(0,\,
0.3^2 I_5)$.  The standard deviation of $0.3$ ensures that in our
application the acceptance rate is about 0.36.

By construction, the Metropolis--Hastings algorithm satisfies detailed balance \eqref{eq:dbc}. \citet{vats:fleg:jone:moa:2015} established that this Markov chain is
geometrically ergodic and that the posterior has a moment generating
function and hence a CLT as at \eqref{eq:multiCLT}
holds.

\subsection{Bayesian one-way random effects model}
\label{sub:eg2}

Suppose for $i \in \{1,\ldots,K\}$,
\[
Y_i \mid \theta_i,\gamma_i\overset{ind.}{\sim} \mathcal{N}(\theta_i,\gamma_i^{-1}),
\quad 
\theta_i \mid \mu,\lambda_\theta,\lambda_i \overset{ind.}{\sim}
\mathcal{N}(\mu,\lambda_\theta^{-1} \lambda_i^{-1}),
\]
$$
\mu \sim \mathcal{N}(m_0,v_0^{-1}), \quad \gamma_i \overset{iid}{\sim} \mathcal{G}(a_3,b_3),\quad
\lambda_\theta \sim \mathcal{G}(a_1,b_1), \quad \lambda_i
\overset{iid}{\sim} \mathcal{G}(a_2,b_2),
$$
where we assume the $a_1,a_2,a_3,b_1,b_2,b_3$ and $v_0$ are known
positive constants while $m_0$ is a known scalar. We consider a data
set simulated under the settings $K=2$, $a_1=a_2=b_1=b_2=0.1$, $a_3=b_3=1.5$,
$m_0=0$ and $v_0=0.001$. Let $y$ denote all of the data, $\lambda=(\lambda_1,\ldots,
\lambda_K)^\top$, $\xi=(\theta_1,\ldots,\theta_K, \mu)^\top$, and
$\gamma=(\gamma_1,\ldots, \gamma_K)^\top$. The hierarchy results in a
proper posterior density $f(\xi, \lambda_\theta, \lambda, \gamma \mid y)$
on $\mathbb{R}^{K+1} \times \mathbb{R}^{2K+1}_{+}$. One can verify that the posterior distribution has a finite second moment.

The posterior is intractable in the sense that posterior expectations
are not generally available in closed form.  We will use a random scan
Gibbs sampler having the posterior as its invariant distribution to
estimate the posterior expectation of all parameters.
\citet{DossHobert2010} derived the full conditional densities
$f(\lambda_\theta \mid \xi,\lambda,\gamma)$,
$f(\lambda \mid \xi,\lambda_\theta,\gamma)$,
$f(\gamma \mid \xi,\lambda_\theta,\lambda)$, and
$f(\xi \mid \lambda_\theta,\lambda,\gamma)$ required to implement random
scan Gibbs.

It is well known that the random scan Gibbs sampler kernel is reversible, namely, satisfies detailed balance \eqref{eq:dbc}, with respect to the posterior; see e.g.,
\citet{robe:rose:2004}. \citet{JohnsonJones2014} established geometric ergodicity of
the random scan Gibbs sampler when $2a_1+K-2>0$ and $a_3>1$. These conditions combined with the second moment condition establish a Markov chain CLT.

\subsection{Multivariate AR(1) process}
\label{sub:AR1}
Consider an AR(1) process $\{X_n;n\in \mathbb{N}\}$ taking values in $\mathbb{R}^p$, i.e.,
$X_{n+1}=AX_n+U_{n+1}$, where $U_n$'s are iid $\mathbb{R}^p$-valued random variables and $A$ is a $p\times p$ matrix.

\citet{Osawa1988} proved that when $U_n$'s follow a normal distribution
$\mathcal{N}_p(\theta,V)$, then this $\mathbb{R}^p$-valued AR(1)
process satisfies detailed balance \eqref{eq:dbc} if and only if the
matrix $AV$ is symmetric. Suppose further that
$\lim_{n\rightarrow\infty}A^n=0$, then it has the stationary
distribution $\mathcal{N}_p [(I-A)^{-1}\theta,(I-A^2)^{-1}V]$. It is
easy to verify that the second moment is finite.

Under stationarity one can derive the lag $t$ autocovariance, $\gamma_t=A^{2t}(I-A^2)^{-1}V$, and hence the covariance matrix, $\Sigma=\{2(I-A^2)^{-1}-I\}(I-A^2)^{-1}V$, as in \eqref{eq:Sigma_inf}. Noticing that $\Sigma$ is finite, and that the Markov chain is reversible with a finite second moment, we establish a Markov chain CLT \eqref{eq:multiCLT} with mean $\mu=(I-A)^{-1}\theta$ and covariance matrix $\Sigma$ \cite[][Corollary 6]{HaggstromRosenthal2007}. Also notice that $\Sigma_0=\gamma_0+2\gamma_1$ is always positive definite, which satisfies the assumption in Remark~\ref{rm:asp} and hence guarantees the asymptotic properties of our proposed estimation method.

Let us consider the following choices that satisfy the conditions above: $\theta=\textbf{1}_p$, $V=I_p$, and $A=p^{-1}H_p\mathrm{diag}(2^{-1}$, $\ldots,2^{-p})H_p^\top$, where $H_p$ is a Hadamard matrix of order $p$. We set $p=12$ in our simulation study.
\subsection{Results}
\label{sub:results}

In this section we refer to the setting of Section~\ref{sub:eg1} as
Example~1, the setting of Section~\ref{sub:eg2} as Example~2, and the setting of Section~\ref{sub:AR1} as Example~3. For all examples we ran 2000 independent replications of the Markov chain
for $10^6$ iterations in Examples~1 and 3 and $5\times10^5$ iterations in
Example~2, respectively. We will compare the multivariate
methods---namely mIS, mISadj, and mK---in the context of estimating the
effective sample size. We then turn our attention to the
finite-sample properties of the confidence regions produced by the
multivariate methods, yielding ellipsoidal regions, and Geyer's
univariate uIS for individual components, yielding cube-shaped
regions. To assess coverage probabilities in Examples~1 and 2 we perform an
independent run of length $10^{10}$ of the Markov chain in each
example and declared the sample average over those $10^{10}$
iterations to be the truth, while in Example~3, the true mean is obtained through the closed form expression derived.

\begin{table}[t!]
  \caption {Estimated ESS with standard
    errors. For uIS, only the minimum estimated ESS is reported.}
\label{tab:mESS}
\centering
\begin{tabular}{ l | c c c | c}
\hline 
  & mK & mIS & mISadj & uIS\\ 
\hline
Ex1($\times10^4$) & 5.40 \tiny{(.002)}  & 5.22 \tiny{(.001)} & 5.18 \tiny{(.001)} & 3.95 {\tiny(.002)}\\
Ex2($\times10^4$) & 4.74 {\tiny(.007)} & 3.76 \tiny{(.002)} & 3.52 \tiny{(.003)} & 1.30 {\tiny(.001)}\\
Ex3($\times10^5$) & 8.78 {\tiny(.000)} & 8.39 \tiny{(.000)} & 8.30 \tiny{(.001)} & 7.58 {\tiny(.001)}\\
\hline  
\end{tabular}
\end{table}

The results concerning effective sample size of the simulation
experiments are given in Table~\ref{tab:mESS}. Prior to the work of \citet{vats:fleg:jone:moa:2015} it
was standard to report the minimum of the univariate effective sample
size calculated component-wise. This leads to a substantial
underestimate of the effective sample size as can be seen in
Table~\ref{tab:mESS}. In contrast, multivariate error estimation yields more accurate evaluation of the effective sample size. We can approximately
order the multivariate methods in terms of estimated effective sample
size: mK $>$ mIS $>$ mISadj. That is, mK is more
optimistic than mIS and mISadj. 

We construct $90\%$ confidence regions using the multivariate
estimation methods and uIS. Throughout ``uIS" and ``uIS-Bonferroni"
represent the uncorrected and Bonferroni corrected confidence regions
generated by uIS, respectively. Let us first examine the volumes of
the confidence regions generated by different methods.

\begin{table}[h!]
\caption {Average volumes to the $p$th ($p=5,8,12$ for Ex1, 2, 3) root and standard errors of nominal $90\%$ confidence regions.}
\label{tab:volume}
\centering
\begin{tabular}{ l | c | c c c | c}
\hline 
  & uIS & mK & mIS & mISadj & uIS-Bonferroni\\ 
\hline
Ex1($\times10^{-3}$) & 5.53 \tiny{(.001)}& 6.31 \tiny{(.001)} & 6.41 \tiny{(.001)} & 6.44 \tiny{(.001)} & 7.82 \tiny{(.001)} \\
Ex2($\times10^{-2}$) & 3.51 \tiny{(.002)} & 3.95 \tiny{(.003)} & 4.43 \tiny{(.003)} & 4.58 \tiny{(.003)} & 5.33 \tiny{(.004)}\\
Ex3($\times10^{-3}$) & 3.84 \tiny{(.000)} & 4.78 \tiny{(.000)} & 4.89 \tiny{(.000)} & 4.92 \tiny{(.000)} & 6.16 \tiny{(.000)}\\
\hline  
\end{tabular}
\end{table}

The volumes are presented in ascending order from left to right across
Table~\ref{tab:volume}. The uncorrected uIS confidence regions are
much smaller than the other methods, while the Bonferroni correction
considerably enlarges the confidence regions, resulting in bigger
volumes than all the multivariate methods.

Recall that the volume of a confidence region depends on the estimated covariance
matrix only through the estimated generalized variance of the Monte Carlo error. Therefore, Table~\ref{tab:volume} compares the estimation of the generalized variance by different multivariate methods. We observe that mK underestimates the generalized variance relatively to mIS. The mISadj method is comparable to mIS in Examples~1 and 3 but clearly overestimates in Example~2.

\begin{table}[ht]
  \caption {Estimated coverage probabilities and standard errors of
    nominal $90\%$ confidence regions.}
\label{tab:coverage}
\centering
\begin{tabular}{ l | c | c c c | c}
\hline 
  & uIS & mK & mIS & mISadj & uIS-Bonferroni\\ 
\hline
Ex1 & .622 \tiny{(.0108)} & .885 \tiny{(.0071)} & .898 \tiny{(.0068)} & .900 \tiny{(.0067)} & .908 \tiny{(.0065)}\\
Ex2 & .386 \tiny{(.0109)} & .660 \tiny{(.0106)} & .845 \tiny{(.0081)} & .881 \tiny{(.0073)} & .862 \tiny{(.0077)}\\
Ex3 & .323 \tiny{(.0105)} & .882 \tiny{(.0072)} & .911 \tiny{(.0064)} & .916 \tiny{(.0062)} & .917 \tiny{(.0062)}\\
\hline  
\end{tabular}
\end{table}

Table~\ref{tab:coverage} shows the empirical coverage probabilities of the confidence regions produced by different methods. The proposed method, mIS, exceeds mK in both the volume and the coverage of confidence regions, although the coverage rate does not always reach the expectation. The adjustment moderately increases the coverage probability.

The uncorrected uIS regions have a poor coverage. The Bonferroni
regions work well in these examples, but in high-dimensional cases the
Bonferroni correction can be overly conservative. Overall,
multivariate error estimation methods yield better confidence regions.

\subsection{A meta-analysis example}
\label{sub:real}

\citet{DossHobert2010} carried out meta-analyses in order to study the effect of non-steroidal anti-inflammatory drugs (NSAIDs) on the risk of colon cancer. The dataset consists of 21 studies that relate NSAIDs intake and risk of colon cancer; see \citet{Harris2005} and \citet{DossHobert2010} for details. We apply the Bayesian one-way random effects model described in Section~\ref{sub:eg2} to the colon cancer dataset. The posterior $f(\theta_1,\ldots,\theta_K,\mu, \lambda_\theta, \lambda_1,\ldots,\lambda_K, \gamma_1,\ldots,\gamma_K \mid y)$ has dimension $p=65$ when $K=21$.

We run a Markov chain for $4\times 10^6$ iterations and compute the multivariate
estimators---namely mIS, mISadj, and mK---along with Geyer's uIS for individual components.

The estimated generalized variances are reported in Table~\ref{tab:realGvar}. The result agrees with our conclusion from the previous simulation study: mISadj is more conservative than mIS; mK clearly underestimates the generalized variance.

\begin{table}[h!]
  \caption {Estimated asymptotic generalized variances ($\times 10^{79}$) of the Monte
    Carlo errors using the colon cancer dataset.}
\label{tab:realGvar}
\centering
\begin{tabular}{c c c}
\hline 
 mK & mIS & mISadj\\ 
\hline
.044 & 6.285 & 77.144 \\
\hline  
\end{tabular}
\end{table}

Table~\ref{tab:realESS} shows the estimated effective sample sizes. The uIS method results in 65 estimated effective sample sizes, each of which corresponds to a component of the posterior distribution. Only the minimum estimated univariate effective sample size is reported.

\begin{table}[h!]
  \caption {Estimated ESS ($\times10^5$) with Monte Carlo sample size $4\times10^6$ using the colon cancer dataset. For uIS, only the minimum estimated ESS is reported.}
\label{tab:realESS}
\centering
\begin{tabular}{c c c| c}
\hline 
 mK & mIS & mISadj & uIS\\ 
\hline
4.637 & 4.296 & 4.134 & 1.137\\
\hline  
\end{tabular}
\end{table}

An advantage of using multivariate methods like mIS over univariate estimation like uIS is that only multivariate methods capture the cross-correlation between components. This cross-correlation is often significant as seen in Figure~\ref{fig:ccf}. 

\begin{figure}[t!]
    \centering
    \includegraphics[width=0.4\linewidth]{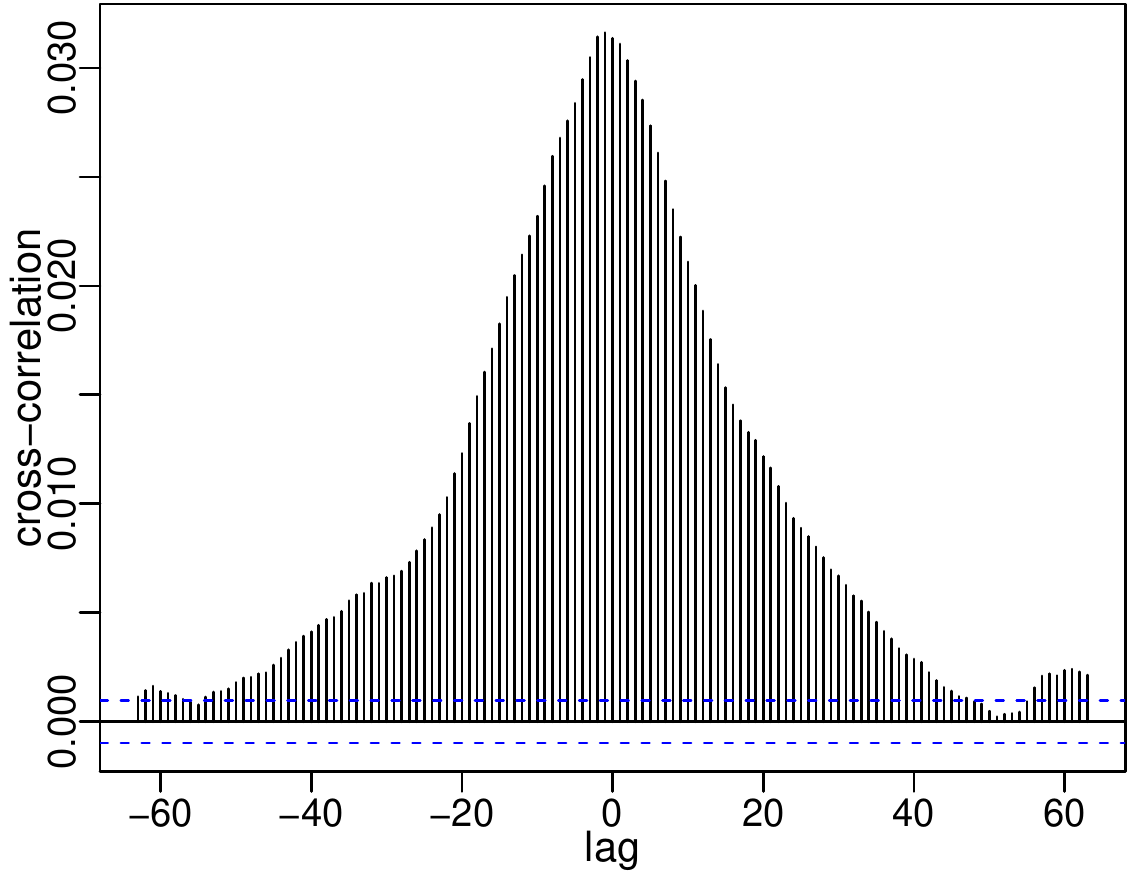}
    \caption{Cross-correlation plot between $\mu$ and $\lambda_1$. Computed with Monte Carlo sample size $4\times10^6$ using the colon cancer dataset.}
    \label{fig:ccf}
\end{figure}

We construct $90\%$ confidence regions using the multivariate
estimation methods and uIS. The left panel of Figure~\ref{fig:realCR}
shows the cross-sections of the confidence regions that are cut
through the center of the confidence regions parallel to the plane
spanned by $\mu$ and $\lambda_1$.  The reader should not be worried
that the cross-sectioned ellipsoids appear much larger than the
Bonferroni region.  

The full 65-dimensional ellipsoid will have a
smaller volume than the 65-dimensional Bonferroni region, but this
does not have to be the case for cross-sectioned regions.  As a
comparison, in the right panel of Figure~\ref{fig:realCR} we present
bivariate $90\%$ confidence regions for $\mu$ and $\lambda_1$ when we
ignore the other 63 components. This clearly shows how multivariate
estimation methods generate confidence regions that are not so liberal
as uIS, yet not so conservative as uIS-Bonferroni.

\begin{figure}[ht!]
\begin{center}
  \begin{subfigure}[ht!]{.4\textwidth}
        \centering
        \includegraphics[width=\linewidth]{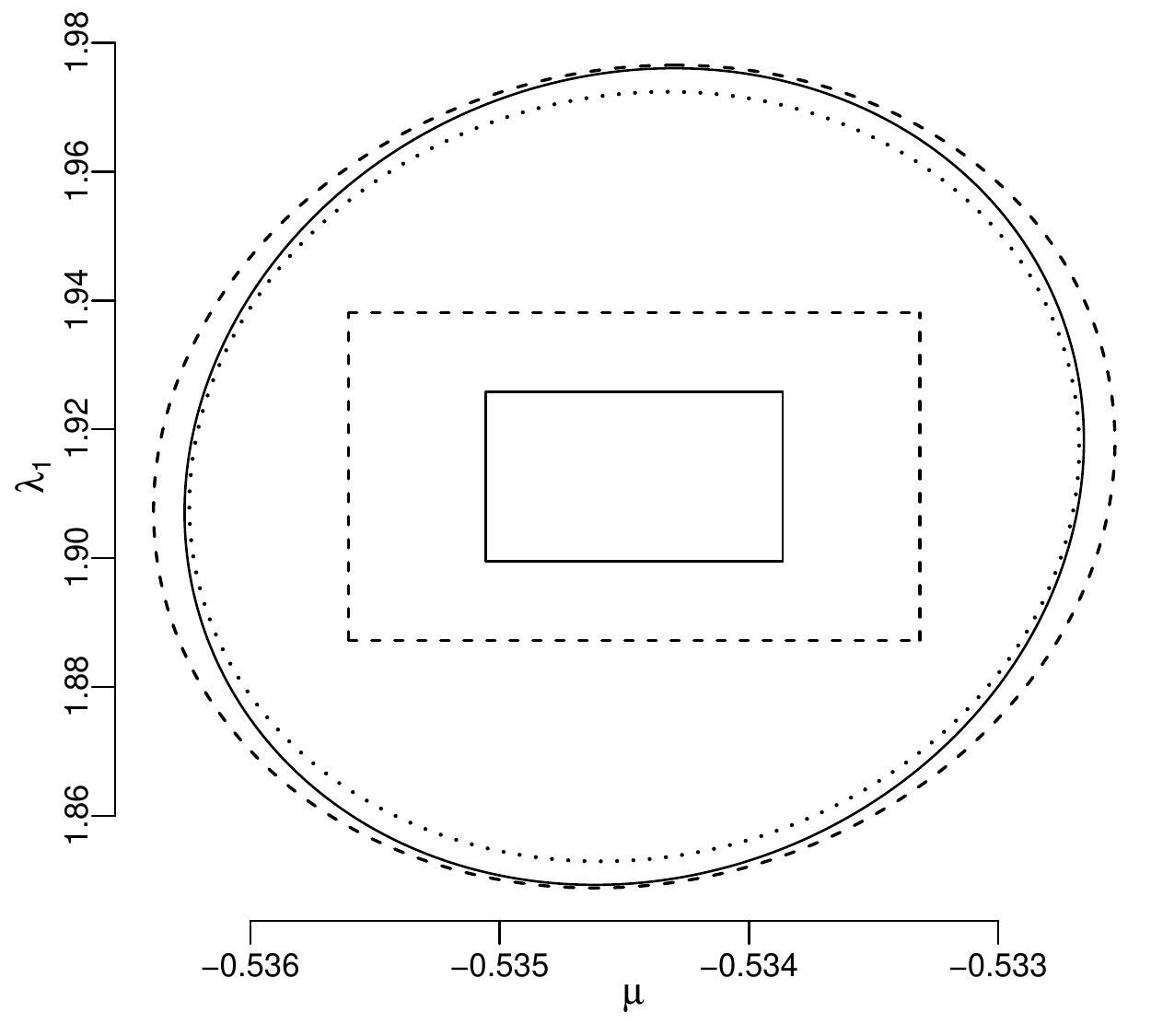}
        \label{fig:rCRcross}
    \end{subfigure}%
    ~ 
    \begin{subfigure}[ht!]{.4\textwidth}
        \centering
        \includegraphics[width=\linewidth]{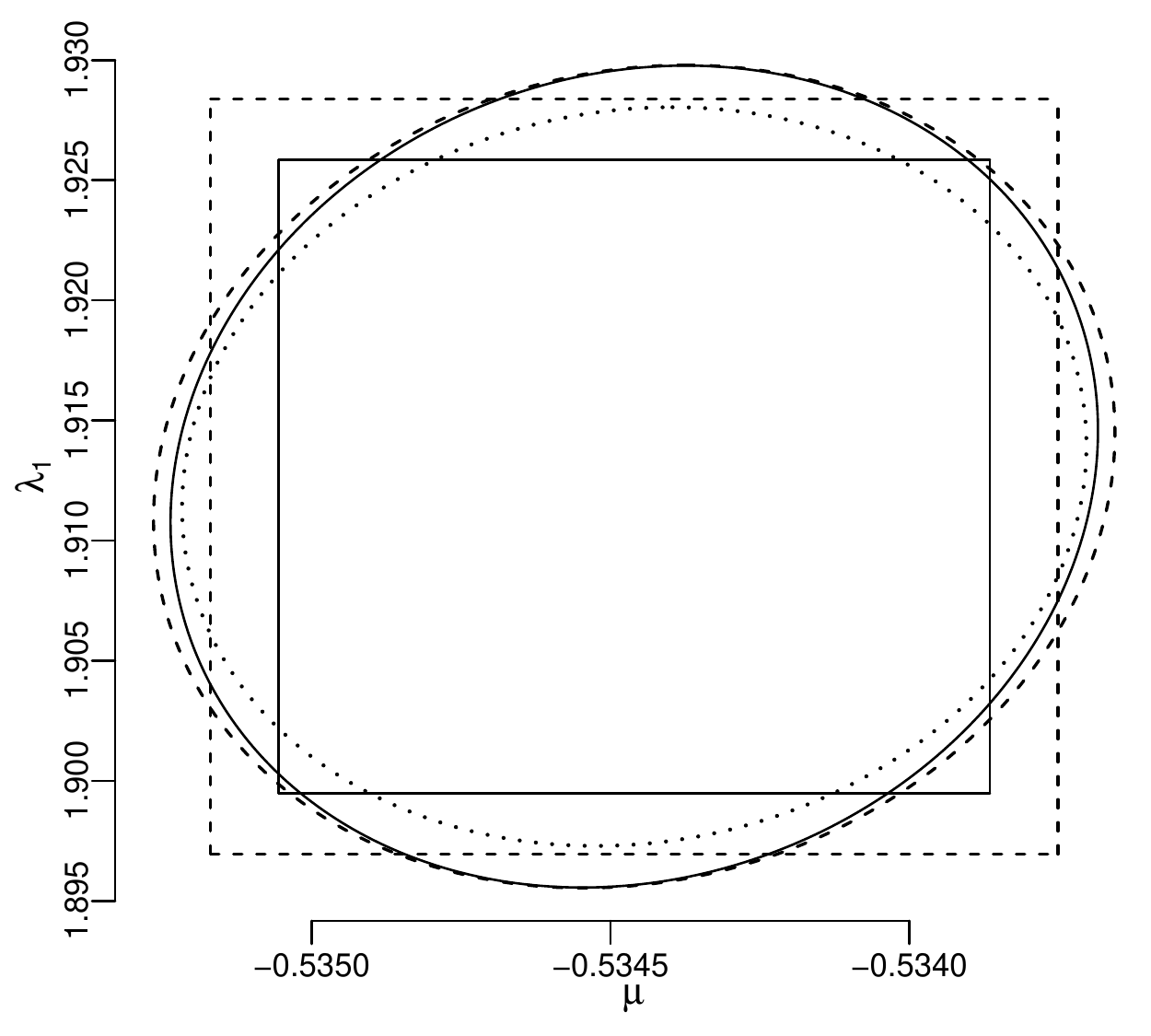}
        \label{fig:realCRsub}
    \end{subfigure}
    \end{center}
    \caption{$90\%$ confidence regions for $(\mu, \,\lambda_1)$
      computed with Monte Carlo sample size $4\times10^6$ using the
      colon cancer dataset. The left panel displays confidence regions
      based on a cross-section of the $p=65$-dimensional region
      parallel to the plane spanned by $(\mu, \,\lambda_1)$. In the
      right panel the confidence regions are created by ignoring the
      other 63 components.  In both panels, the solid ellipsoid,
      dashed ellipsoid, and dotted ellipsoid corresponds to mIS,
      mISadj, and mK, respectively while the small solid and the big
      dashed rectangles are uIS and uIS-Bonferroni, respectively.  }
    \label{fig:realCR}
\end{figure}

\begin{table}[ht!]
\caption {Volumes to the $p$th root ($\times 10^{-3}$) of $90\%$ confidence regions for all components ($p=65$) and for $\mu$ and $\lambda_1$ only ($p=2$). Computed with Monte Carlo sample size $4\times10^6$ using the colon cancer dataset.}
\label{tab:realvolume}
\centering
    \begin{subtable}{.5\textwidth}
        \centering
        \begin{tabular}{ c | c c c | c}
        \hline 
        uIS & mK & mIS & mISadj & uIS-Bonferroni\\ 
        \hline
        6.96 & 8.70 & 9.04 & 9.22 & 13.39\\
        \hline  
        \end{tabular}
    \caption{65-dimensional confidence regions.}\label{tab:realvolume65}
    \end{subtable}
   ~
    \begin{subtable}{.5\textwidth}
            \centering
        \begin{tabular}{ c | c c c | c}
        \hline 
        uIS & mK & mIS & mISadj & uIS-Bonferroni\\ 
        \hline
        5.60 & 6.04 & 6.44 & 6.56 & 6.68\\
        \hline  
        \end{tabular}
    \caption{Bivariate confidence regions for $\mu$ and $\lambda_1$.}\label{tab:realvolume2}
    \end{subtable}
 
\end{table}

Table~\ref{tab:realvolume} compares the volumes of the confidence
regions generated by different methods. The results agree with our
conclusion from the previous simulation study: mISadj is slightly more
conservative than mIS; mK clearly underestimates the generalized
variance. The volumes generated by multivariate estimators are fairly
close to each other but the univariate results are far
away. Apparently uIS is too liberal while uIS-Bonferroni is too
conservative, but the multivariate methods achieve a balance.

\subsection{Discussion}
\label{sub:discussion}
The preceding simulation examples and the theory developed indicate
that mIS and mISadj perform as they were designed to in that they
provide a consistent overestimate of the asymptotic generalized
variance of the Monte Carlo error. Compared to standard univariate
methods, our estimators adjust for multivariate issues and thus provide
more realistic estimates of Monte Carlo effective sample size and
slightly larger confidence regions which result in improved
performance in terms of coverage probabilities.

\bigskip
\noindent
\textbf{Acknowledgments}. The authors are grateful to Charles Geyer and Dootika Vats for helpful
conversations.

\appendix
\section*{Appendices}
\addcontentsline{toc}{section}{Appendices}
\renewcommand{\thesubsection}{\Alph{subsection}}

\subsection{Proofs of Propositions 1 and 2}
\label{app:propositions}
We begin with some preliminary results which will be useful later.

\begin{lemma} [\citet{Harville1997}, Lemma 18.2.17]\label{lemma:harville}
Let $A_0,A_1,A_2,\ldots$ represent a sequence of $m\times n$ matrices. If the infinite series $\sum_{k=0}^\infty A_k$ converges, then $\lim_{k\to\infty}A_k=0$.
\end{lemma}

Since the eigenvalues of a Hermitian $p\times p$ matrix $A$ are real,
we may (and do) adopt the convention that they are always arranged in
algebraically non-decreasing order:
\begin{equation}\label{eq:eigen_order}
\lambda_{\mathrm{min}}(A)=\lambda^{(1)}(A)\leq
\lambda^{(2)}(A)\leq\cdots\leq\lambda^{(p-1)}(A)\leq\lambda^{(p)}(A) =
\lambda_{\mathrm{max}}(A). \tag{A.1}
\end{equation}

\begin{lemma} [\citet{HornJohnson1985}, Corollary 4.3.15]
  Let $p\times p$ matrices $A$, $B$ be Hermitian and let the
  respective eigenvalues of $A$, $B$, and $A+B$ be
  $\{\lambda^{(k)}(A): k=1,\ldots,p\}$, $\{\lambda^{(k)}(B): k=1,\ldots,p\}$, and
  $\{\lambda^{(k)}(A+B): k=1,\ldots,p\}$, each algebraically ordered as in
  \eqref{eq:eigen_order}. Then, for all $k\in \{1,\ldots,p\}$,
\begin{equation}\label{eq:HornJohnson}
\lambda^{(k)}(A)+\lambda^{(1)}(B)\leq\lambda^{(k)}(A+B).
\tag{A.2}
\end{equation}
\end{lemma}

\begin{lemma}\label{lemma:matsorting}
  Suppose we have two $p\times p$ Hermitian matrices $A$ and $B$. Let
  the respective eigenvalues of $A$ and $B$ be
  $\{\lambda^{(k)}(A): k=1,\ldots,p\}$ and $\{\lambda^{(k)}(B): k=1,\ldots,p\}$, each
  algebraically ordered as in \eqref{eq:eigen_order}. If $A-B$ is
  positive definite, then $\lambda^{(k)}(A)>\lambda^{(k)}(B)$, for all
  $k \in \{1,\ldots,p\}$. Further, if $A$ and $B$ are both positive
  semi-definite, then $|A|>|B|$.
\end{lemma}
\begin{proof}
Applying \eqref{eq:HornJohnson} to $B$ and $A-B$, we get for all $k \in \{1,\ldots,p\}$,
$
\lambda^{(k)}(B)+\lambda^{(1)}(A-B)\leq\lambda^{(k)}(A).
$
Since $A-B$ is positive definite, we have $\lambda^{(1)}(A-B)>0$.
Therefore, for all $k \in \{ 1,\ldots,p\}$, $\lambda^{(k)}(A)>\lambda^{(k)}(B)$.
When $A$ and $B$ are both positive semi-definite, we further have for all $k \in \{1,\ldots,p\}$,
\begin{equation}\label{eq:ineq_eigen}
\lambda^{(k)}(A)>\lambda^{(k)}(B)\geq0. \tag{A.3}
\end{equation}
Since the determinant is equal to the product of all eigenvalues, we take product of \eqref{eq:ineq_eigen} for all $k \in \{ 1, \ldots, p\}$ and obtain $|A|>|B|$.
\end{proof}

\begin{lemma}[\citet{vats:fleg:jone:sve:2016}, Theorem 2]\label{lemma:sm_eigen}
Let $\Sigma_n$ be a strongly consistent estimator of $\Sigma$. Let
  the respective eigenvalues of $\Sigma_n$ and $\Sigma$ be
  $\{\lambda^{(k)}(\Sigma_n): k=1,\ldots,p\}$ and $\{\lambda^{(k)}(\Sigma) : k=1,\ldots,p\}$, each
  algebraically ordered as in \eqref{eq:eigen_order}. Then $\lambda^{(k)}(\Sigma_n)\to\lambda^{(k)}(\Sigma) ~ \text{with probability}~1$ as $n\to\infty$ for all $k \in \{ 1,\ldots,p\}$.
\end{lemma}
\begin{corollary}\label{corol:matconverge}
Let $\Sigma_n$ be a strongly consistent estimator of $\Sigma$,
then $|\Sigma_n|\to|\Sigma| ~ \text{with probability}~1$ as $n\to\infty$.
\end{corollary}

\subsubsection{Proof of Proposition 1}
\label{appsub:prop1}

We begin with the univariate case so $g:\mathcal{X}\to\mathbb{R}$. Let
$\mathcal{E}$ be the spectral decomposition measure associated with
transition kernel $P$ and $\mathcal{E}_g$ be the induced spectral
measure for $g$. Details on the spectral decomposition measure can be found in \citet{rudi:1991}, \citet{chan:geye:1994}, and \citet{HaggstromRosenthal2007}. Specifically, for all $t\in\mathbb{N}$,
\begin{equation}\label{eq:gam_t}
\gamma_t=\int_{-1}^1\lambda^t \mathcal{E}_g(d\lambda). \tag{A.4}
\end{equation}
It follows that for all $i\in\mathbb{N}$,
\[
\Gamma_i=\gamma_{2i}+\gamma_{2i+1} =
\int_{-1}^1\lambda^{2i}(1+\lambda) \mathcal{E}_g(d\lambda) 
\]
and
\[
\Gamma_i-\Gamma_{i+1} =
\int_{-1}^1\lambda^{2i}(1+\lambda)^2(1-\lambda)
\mathcal{E}_g(d\lambda) .
\]
Therefore, $\Gamma_i$ and $\Gamma_i-\Gamma_{i+1}$ must be non-negative. To prove Proposition \ref{prop:1}(i) that $\Gamma_i>0$ and Proposition \ref{prop:1}(ii) that $\Gamma_i-\Gamma_{i+1}>0$, we need to show that neither $\Gamma_i$ nor $\Gamma_i-\Gamma_{i+1}$ can be zero. For $i=0$,
\begin{equation}\label{eq:i=0,G}
\Gamma_0=\int_{-1}^1(1+\lambda)\mathcal{E}_g(d\lambda)=0  \quad \Leftrightarrow \quad \mathcal{E}_g(\{-1\})=1, \tag{A.5}
\end{equation}
and
\begin{equation}\label{eq:i=0,dG}
\Gamma_0-\Gamma_1 =
\int_{-1}^1(1+\lambda)^2(1-\lambda)\mathcal{E}_g(d\lambda) = 0 \quad \Leftrightarrow \quad \mathcal{E}_g(\{-1,1\}) = 1.\tag{A.6}
\end{equation}
For $i\in\mathbb{N}^+$,
\begin{equation}\label{eq:i>0,G}
\Gamma_i=\int_{-1}^1\lambda^{2i}(1+\lambda) \mathcal{E}_g(d\lambda)=0 \quad \Leftrightarrow \quad  \mathcal{E}_g(\{-1,0\})=1,\tag{A.7}
\end{equation}
and
\begin{equation}\label{eq:i>0,dG}
\Gamma_i-\Gamma_{i+1} =
\int_{-1}^1\lambda^{2i}(1+\lambda)^2(1-\lambda)\mathcal{E}_g(d\lambda)=0 \quad \Leftrightarrow \quad \mathcal{E}_g(\{-1,0,1\}) = 1.\tag{A.8}
\end{equation}

By \eqref{eq:i=0,G}--\eqref{eq:i>0,dG}, for an arbitrary $i\in\mathbb{N}$, a necessary condition for each of $\Gamma_i=0$ and $\Gamma_i-\Gamma_{i+1}=0$ is $\mathcal{E}_g(\{-1,0,1\}) = 1$. We now show that $\mathcal{E}_g(\{-1,0,1\}) = 1$ cannot hold under our assumptions, so that both $\Gamma_i$ and $\Gamma_i-\Gamma_{i+1}$ are non-zero, which completes the proof of Proposition \ref{prop:1}(i)--(ii).

If $\mathcal{E}_g$ is a point mass at 0, then \eqref{eq:gam_t} yields
\[
\gamma_t=\int_{-1}^1\lambda^t \mathcal{E}_g(d\lambda)=0
\]
for all $t\in\mathbb{N}^+$, which is trivial. Therefore, without loss of generality, we assume
\begin{equation}\label{eq:notall0}
\mathcal{E}_g(\{0\})<1. \tag{A.9}
\end{equation}
\citet{HaggstromRosenthal2007} showed that when $P$ is irreducible and aperiodic,
\begin{equation}\label{eq:supp-1,1}
\mathcal{E}_g(\{-1,1\})=0. \tag{A.10}
\end{equation}

It follows from \eqref{eq:notall0} and \eqref{eq:supp-1,1} that $\mathcal{E}_g(\{-1,0,1\})<1$. By previous arguments, we have proved Proposition~\ref{prop:1}(i)-(ii). That
is, for all $i\in\mathbb{N}$, $\Gamma_i>0$ and $\Gamma_i-\Gamma_{i+1}>0$.

Proposition~\ref{prop:1}(iii), namely $\lim_{i\to\infty}\Gamma_i=0$, follows from Lemma~\ref{lemma:harville} and the assumption that $\sum_{i=0}^\infty\Gamma_i$ converges. 

Finally, by Proposition~\ref{prop:1}(i)--(iii), we obtain Proposition~\ref{prop:1}(iv), i.e., $\{\Gamma_i: i\in\mathbb{N}\}$ is positive, decreasing, and converges to $0$.

We now turn to the multivariate case so $g:\mathcal{X}\to\mathbb{R}^p$
and $p\geq2$. Set $h=v^\top g$ for an arbitrary $v\in \mathbb{R}^p$ and
$v\neq0$.  Then $h:\mathcal{X}\to\mathbb{R}$ is measurable and square
integrable with respect to $F$. Recall that the Markov
chain is assumed stationary.  For $t\in\mathbb{N}$ define the lag $t$
autocovariance
\[
\gamma_t^*=\gamma_{-t}^*=\mathrm{cov}_F\{h(X_i),h(X_{i+t})\}
\]
and for $i\in\mathbb{N}$ define
$\Gamma_i^*=\gamma_{2i}^*+\gamma_{2i+1}^* $.
Notice that
\begin{align*}
\Gamma_i^*&=\gamma_{2i}^*+\gamma_{2i+1}^* 
=\mathrm{cov}_F\{ h(X_0),h(X_{2i})\} +\mathrm{cov}_F\{ h(X_0),h(X_{2i+1})\} \\
&=v^\top\mathrm{cov}_F \{ g(X_0),g(X_{2i})\} v+v^\top\mathrm{cov}_F\{ g(X_0),g(X_{2i+1})\} v\\
&=v^\top\gamma_{2i}v+v^\top\gamma_{2i+1}v
=v^\top(\gamma_{2i}+\gamma_{2i+1})v
=v^\top\Gamma_iv .
\end{align*}
By the univariate case considered above, $\Gamma_i^* >0$. Since $v$ is arbitrary, $\Gamma_i$ is positive definite. A similar argument shows that $\Gamma_i-\Gamma_{i+1}$ is positive
definite. This establishes Proposition~\ref{prop:1}(i)--(ii).

Use Lemma~\ref{lemma:harville} and notice that $\sum_{i=0}^\infty\Gamma_i$ converges by assumption.
We obtain $\lim_{i\to\infty}\Gamma_i=0$. Thus Proposition~\ref{prop:1}(iii) is proved.

Since $\Gamma_i$ is positive definite, $\xi_i>0$ for all
$i\in\mathbb{N}$.  Since $\Gamma_i-\Gamma_{i+1}$ is positive definite
we obtain from Lemma~\ref{lemma:matsorting} that $\xi_i>\xi_{i+1}$.
Hence $\xi_i \to 0$ as $i \to \infty$ which establishes
Proposition~\ref{prop:1}(iv). \hfill $\Box$

\subsubsection{Proposition 2}
\label{appsub:prop2}
For all $m\in\mathbb{N}$ let $\lambda_m$ be the smallest eigenvalue of
$\Sigma_m$.  Notice that $\Sigma_m-\Sigma_{m-1}=2\Gamma_m$ is positive
definite by Proposition~\ref{prop:1}. Then
Lemma~\ref{lemma:matsorting} implies $\lambda_m>\lambda_{m-1}$ and
hence$\{\lambda_m: m\in\mathbb{N}\}$ is monotonically increasing. Since
$\{\Sigma_m: m\in\mathbb{N}\}$ converges to the asymptotic covariance
matrix $\Sigma$, by Lemma \ref{lemma:sm_eigen} we have $$\lim_{m\to
  \infty}\lambda_m=\lambda>0,$$ where $\lambda$ is the smallest
eigenvalue of $\Sigma$.

If $\lambda_0\leq 0$, there exists a positive integer $m_0$ such that
$\lambda_m>0$ for $m\geq m_0$ and $\lambda_m\leq0$ for $m<m_0$. If
$\lambda_0>0$, then $\lambda_m>0$ for all $m\in\mathbb{N}$. In this
case, let $m_0=0$. Immediately we have that $\Sigma_m$ is positive
definite for $m\geq m_0$ and not positive definite for $m<m_0$. It
then follows that for all $m\geq m_0$, $|\Sigma_m|>0$.

Now let $m> m_0$ and notice that $\Sigma_m-\Sigma_{m-1}$ is positive
definite. Using Lemma~\ref{lemma:matsorting} we obtain for all $m>
m_0$, 
$
|\Sigma_m|>|\Sigma_{m-1}| .
$
By $\lim_{m\to\infty}\Sigma_m=\Sigma$ and Corollary
\ref{corol:matconverge}, 
\[
\lim_{m\to\infty} |\Sigma_m|=|\Sigma|.
\]
Therefore, $\{|\Sigma_m|:m=m_0,m_0+1,m_0+2,\ldots\}$ is positive, increasing, and converges to $|\Sigma|$. \hfill $\Box$

\subsection{Proofs of Theorems 1 and 2}
\label{app:theorems}
\begin{lemma} \label{lemma:cst_gamma} For all $t\in\mathbb{N}$, with
  probability 1, as $n \to \infty$, 
$
\gamma_{n,t} \to \gamma_t .
$
\end{lemma}
\begin{proof}
Notice that 
\begin{align*}
\gamma_{n,t}=&\frac{1}{n} \sum\limits_{i=1}^{n-t}\left\{g(X_{n,i})-\mu_n\right\}\left\{g(X_{n,i+t})-\mu_n\right\}^\top\\
=&\frac{1}{n}
\sum\limits_{i=1}^{n-t}g(X_{n,i})g(X_{n,i+t})^\top-\frac{1}{n}
\sum\limits_{i=1}^{n-t}g(X_{n,i})\mu_n^\top-
\frac{1}{n} \, \mu_n\sum\limits_{i=1}^{n-t}g(X_{n,i+t})^\top+
\frac{n-t}{n} \, \mu_n\mu_n^\top \; .
\end{align*}

By repeated application of the Markov chain strong law we see that,
with probability 1, as $n \to \infty$, 
\[
\gamma_{n,t} \to \mathrm{E}_F \{g(X_0)g(X_t)^\top \}-\mu\mu^\top=\mathrm{cov}_F\left\{ g(X_0),g(X_t)\right\} =\gamma_t
.
\]
\end{proof}

\begin{corollary} \label{corol:cst_Sigmai} For all $m\in\mathbb{N}$, with
  probability 1, as $n \to \infty$,
$
\Sigma_{n,m} \to \Sigma_m.
$
\end{corollary}
\begin{proof}
This follows immediately from Lemma \ref{lemma:cst_gamma}.
\end{proof}

\begin{lemma}\label{lemma:limXA}
If a sequence of random variables $X_1,X_2,\ldots$ converges to $X$ with probability 1, then, for an arbitrary $x\in\mathbb{R}$ such that $\Pr(X=x)=0$,
\begin{equation}\label{eq:leq}
\liminf_{n\rightarrow\infty}\{X_n\leq x\}=\{X\leq x\}~w.p.~1 \tag{A.11}
\end{equation}
and
\begin{equation}\label{eq:>}
\liminf_{n\rightarrow\infty}\{X_n> x\}=\{X>x\}~w.p.~1. \tag{A.12}
\end{equation}
\end{lemma}

\begin{proof}
We only prove the first part. The second part can be shown by a similar argument.

Recall that two events $A$ and $B$ are equal almost surely if both of
the events $A\setminus B$ and $B\setminus A$ are null sets
\cite[][p. 13]{FG1996}.  Thus we need only show that both
$\liminf_{n\rightarrow\infty}\{X_n\leq x\}\setminus\{X\leq x\}$ and
$\{X\leq x\}\setminus\liminf_{n\rightarrow\infty}\{X_n\leq x\}$ are
null sets.

Suppose $\omega\in\liminf_{n\rightarrow\infty}\{X_n\leq
x\}\setminus\{X\leq x\}$. By
definition, 
\[
\omega\in\liminf_{n\rightarrow\infty}\{X_n\leq x\}
\] 
is equivalent to saying that there exists some $n$ such that for all
$m\geq n$, $X_m(\omega)\leq x$. This implies
that 
\[
\lim_{n\rightarrow\infty}X_n(\omega)\leq x<X(\omega),
\] 
where the second inequality is due to $\omega\notin\{X\leq x\}$. It
follows that
\[
\omega\in\Bigl\{\lim_{n\rightarrow\infty}X_n\neq X\Bigr\}.
\] 
Thus we have that
\[
\liminf_{n\rightarrow\infty}\{X_n\leq x\}\setminus\{X\leq x\} \subset
\left\{\lim_{n\rightarrow\infty}X_n\neq X\right\}
\]
which is a null set because $X_n\overset{a.s.}{\rightarrow}X$.

Suppose $\omega\in\{X\leq
x\}\setminus\liminf_{n\rightarrow\infty}\{X_n\leq x\}$. By definition, 
\[ 
\omega\notin\liminf_{n\rightarrow\infty}\{X_n\leq x\}
\] 
is equivalent to saying that for all $n$, there exists some $m\geq n$
such that $X_m(\omega)>x$. This implies that
\[
\lim_{n\rightarrow\infty}X_n(\omega)\geq x\geq X(\omega),
\]
where the second inequality is due to $\omega\in\{X\leq x\}$.  It
follows that 
\[
\omega\in\left\{\lim_{n\rightarrow\infty}X_n\neq X\right\}\bigcup\{X=x\}.
\] 
Thus we have that
\[
\{X\leq x\}\setminus\limsup_{n\rightarrow\infty}\{X_n\leq x\} \subset
\left\{\lim_{n\rightarrow\infty}X_n\neq X\right\}\bigcup\{X=x\},
\]
which is a null set.

So far we have proved \eqref{eq:leq}. A similar argument can be used to prove \eqref{eq:>}.
\end{proof}

Equipped with the preceding results, we now prove the following lemma in preparation for Theorems 1 and 2.

Recall that $m_0$ is a non-negative integer such that $\Sigma_m$ is positive definite for $m\geq m_0$ and not positive definite for $m< m_0$. Also recall that $s_n$ is the smallest integer such that
$\Sigma_{n,s_n}$ is positive definite and that $t_n$ is the largest
integer $m$ ($s_n\leq m\leq \lfloor n/2-1\rfloor$) such that
$|\Sigma_{n,i}|>|\Sigma_{n,i-1}|$ for all $i \in \{s_n+1,\ldots,m\}$. The smallest eigenvalues of $\Sigma_m$ and $\Sigma_{n,m}$ are denoted $\lambda_{m}$ and $\lambda_{n,m}$, respectively.

\begin{lemma}\label{lemma:Indicator0}
Suppose $\liminf_{n\rightarrow\infty}\{\lambda_{n,m_0-1}\leq0\}$
occurs with probability 1.  For all $K\geq m_0$,
\[
\Pr\left(\liminf_{n\to\infty}\{s_n=m_0,t_n\geq K\}\right)=1.
\]
\end{lemma}
\begin{proof}
  Define $\Delta_{m}=|\Sigma_m| -|\Sigma_{m-1}|$ and
  $\Delta_{n,m}=|\Sigma_{n,m}| -|\Sigma_{n,m-1}|$.  Notice that
\begin{align*}
\{s_n=m_0,t_n\geq K\}
&=\{s_n=m_0\} \cap \{ \Delta_{n,i}>0\text{ for all }i
  \text{ such that } m_0<i\leq K\}\\
&=\left(\bigcap_{m<m_0}\{\Sigma_{n,m}\text{ is not positive definite}\}\right)\bigcap\{\Sigma_{n,m_0}\text{ is positive definite}\} 
\bigcap \left(\bigcap_{m_0<i\leq K}\{ \Delta_{n,i}>0\}\right)\\
&=\left(\bigcap_{m<m_0}\{\lambda_{n,m}\leq0\}\right)\bigcap\{\lambda_{n,m_0}>0\}\bigcap \left(\bigcap_{m_0<i\leq K}\{ \Delta_{n,i}>0\}\right),
\end{align*}
where $\lambda_{n,m}$ denotes the smallest eigenvalue of $\Sigma_{n,m}$. Then we write
\begin{align}
& \ \ \ \ \liminf_{n\to\infty}\{s_n=m_0,t_n\geq K\}\nonumber\\
&=\left(\bigcap_{m<m_0}\liminf_{n\to\infty}\{\lambda_{n,m}\leq0\}\right)\bigcap\liminf_{n\to\infty}\{\lambda_{n,m_0}>0\}\bigcap \left(\bigcap_{m_0<i\leq K}\liminf_{n\to\infty}\{ \Delta_{n,i}>0\}\right)\label{eq:intersec}. \tag{A.13}
\end{align}
By Lemma \ref{lemma:sm_eigen}, Corollary \ref{corol:matconverge} and
Corollary \ref{corol:cst_Sigmai}, for all $m$, with probability 1,
\begin{equation}\label{eq:cst_lambda}
\lambda_{n,m} \to \lambda_m\text{ as }n\to\infty, \tag{A.14}
\end{equation}
and for all $i$, with probability 1,
\begin{equation}\label{eq:cst_delta}
\Delta_{n,i} \to \Delta_i\text{ as } n\to\infty. \tag{A.15}
\end{equation}
By Proposition \ref{prop:2}(i), $\lambda_{m_0}>0$ so that
$\Pr(\lambda_{m_0}=0)=0$ and $\lambda_m\leq0$ for all $m<m_0$. In
particular, $\lambda_m<0$ so that $\Pr(\lambda_{m}=0)=0$ for all
$m<m_0-1$. By Proposition \ref{prop:2}(ii), $\Delta_i>0$ so
that $\Pr(\Delta_i=0)=0$ for $i>m_0$. Then by Lemma \ref{lemma:limXA} we
have that for all $m<m_0-1$,
\begin{equation}\label{eq:m0-1}
\liminf_{n\to\infty}\{\lambda_{n,m}\leq0\}\overset{a.s.}{=}\left\{\lim_{n\to\infty}\lambda_{n,m}\leq0\right\}, \tag{A.16}
\end{equation}
\[
\liminf_{n\to\infty}\{\lambda_{n,m_0}>0\}\overset{a.s.}{=}\left\{\lim_{n\to\infty}\lambda_{n,m_0}>0\right\},
\]
and for $i>m_0$
\[
\liminf_{n\to\infty}\{ \Delta_{n,i}>0\}\overset{a.s.}{=}\left\{\lim_{n\to\infty}\Delta_{n,i}>0\right\}.
\]
Notice that \eqref{eq:m0-1} holds for $m=m_0-1$ if
$\lambda_{m_0-1}<0$. When $\lambda_{m_0-1}=0$, \eqref{eq:m0-1} is true
only if $\liminf_{n\rightarrow\infty}\{\lambda_{n,m_0-1}\leq0\}$
occurs with probability 1. 

Under the preceding assumption, we continue to write \eqref{eq:intersec} as
\begin{align}
\liminf_{n\to\infty}\{s_n=m_0,t_n\geq K\} 
&=\left(\bigcap_{m<m_0}\liminf_{n\to\infty}\{\lambda_{n,m}\leq0\}\right)\bigcap\liminf_{n\to\infty}\{\lambda_{n,m_0}>0\}\bigcap \left(\bigcap_{m_0<i\leq K}\liminf_{n\to\infty}\{ \Delta_{n,i}>0\}\right)\nonumber\\
&\overset{a.s.}{=}\left(\bigcap_{m<m_0} \left\{\lim_{n\to\infty}\lambda_{n,m}\leq0\right\}\right)\bigcap \left\{\lim_{n\to\infty}\lambda_{n,m_0}>0\right\} \bigcap \left(\bigcap_{m_0<i\leq K} \left\{ \lim_{n\to\infty}\Delta_{n,i}>0\right\}\right)\label{eq:lim_intersec}. \tag{A.17}
\end{align}
By Proposition \ref{prop:2}(i), $\lambda_{m_0}>0$ and $\lambda_m\leq0$ for all $m<m_0$. Then by \eqref{eq:cst_lambda} we have for $m<m_0$
\begin{equation}\label{eq:1as1}
\Pr\left(\lim_{n\to\infty}\lambda_{n,m}\leq0\right)\geq \Pr\left(\lim_{n\to\infty}\lambda_{n,m}=\lambda_m\right)=1,\tag{A.18}
\end{equation}
and
\begin{equation}\label{eq:2as1}
\Pr\left(\lim_{n\to\infty}\lambda_{n,m_0}>0\right)\geq \Pr\left(\lim_{n\to\infty}\lambda_{n,m_0}=\lambda_{m_0}\right)=1. \tag{A.19}
\end{equation}
By Proposition \ref{prop:2}(ii), $\Delta_i>0$ for $i>m_0$. Then by \eqref{eq:cst_delta} we have
\begin{equation}\label{eq:3as1}
\Pr\left(\lim_{n\to\infty}\Delta_{n,i}>0\right)\geq \Pr\left(\lim_{n\to\infty}\Delta_{n,i}=\Delta_{i}\right)=1. \tag{A.20}
\end{equation}
It follows from \eqref{eq:1as1}--\eqref{eq:3as1} that
\[
\Pr\left\lbrace\left(\bigcap_{m<m_0} \left\{\lim_{n\to\infty}\lambda_{n,m}\leq0\right\}\right)\bigcap \left\{\lim_{n\to\infty}\lambda_{n,m_0}>0\right\}\bigcap \left(\bigcap_{m_0<i\leq K}\left\{ \lim_{n\to\infty}\Delta_{n,i}>0\right\}\right)\right\rbrace=1.
\]
Then by \eqref{eq:lim_intersec} we obtain the result.
\end{proof}

\begin{remark}
  Consider the assumption that
  $\liminf_{n\rightarrow\infty}\{\lambda_{n,m_0-1}\leq0\}$ occurs with
  probability 1. If $m_0=0$, then this assumption is not required for
  the Lemma; recall Remark~\ref{rm:asp}.  In addition, the assumption
  holds if $\Sigma_{m_0-1}$ is not positive semi-definite.  Recall
  from Proposition~\ref{prop:2}(i) we have that $\Sigma_{m_0-1}$ is not
  positive definite but, of course, it may still be positive
  semi-definite.
\end{remark}

\subsubsection{Theorem 1: Feasibility of the estimation method}
\label{appsub:theorem1}
\begin{proof}
When $K\geq m_0$ and $n>2m_0$, 
$
\{s_n\text{ exists}\}\supset\{s_n=m_0\}\supset\{s_n=m_0,t_n\geq K\}.
$
Then the result follows from Lemma~\ref{lemma:Indicator0}.
\end{proof}

\subsubsection{Theorem 2: Overestimation for the Asymptotic Generalized Variance of the Monte Carlo Error}
\label{appsub:theorem2}
\begin{proof}
We need to prove, for all $\epsilon>0$,
\begin{equation}\label{eq:liminfty}
\Pr\left(\bigcap_{n=N}^\infty\{|\Sigma_{n,t_n}|>|\Sigma|-\epsilon\}\right)\to1 \text{ as }N\to\infty. \tag{A.21}
\end{equation}
Recall that $\Delta_i$ is defined as $\Delta_i=|\Sigma_i|-|\Sigma_{i-1}|$.

By Proposition \ref{prop:2}(ii) that $\lim_{m\to\infty}|\Sigma_m|=|\Sigma|$, we can write
$$\sum\limits_{i=m_0+1}^\infty\Delta_{i}=|\Sigma|-|\Sigma_{m_0}|<\infty,$$
so $\sum\limits_{i=m_0+1}^\infty\Delta_{i}$ converges; and hence the tail must converge to 0.
Therefore, for all $\epsilon>0$, there exists $K_\epsilon\geq m_0$ such that
\begin{equation}\label{eq:Sigmatailbdd}
|\Sigma|-|\Sigma_{K_\epsilon}|=\sum_{i=K_\epsilon+1}^\infty\Delta_{i}<\epsilon/2. \tag{A.22}
\end{equation}
Notice that
\begin{align}
 \left\{|\Sigma_{n,t_n}|>|\Sigma|-\epsilon\right\} 
\supset\left\{|\Sigma_{n,t_n}|>|\Sigma|-\epsilon\right\}\bigcap\left\{s_n=m_0,t_n\geq K_\epsilon\right\} \supset\left\{|\Sigma_{n,K_\epsilon}|>|\Sigma|-\epsilon\right\}\bigcap\left\{s_n=m_0,t_n\geq K_\epsilon\right\}\label{eq:intersec2}. \tag{A.23}
\end{align}
The second step in \eqref{eq:intersec2} is due to the definition of mIS:
\begin{center}
``$s_n=m_0$ and $t_n\geq K_\epsilon$ for some $K_\epsilon\geq m_0$" implies ``$|\Sigma_{n,t_n}|\geq|\Sigma_{n,K_\epsilon}|$".
\end{center}
It follows directly from \eqref{eq:intersec2} that
\begin{align*}
 \bigcap_{n=N}^\infty\left\{|\Sigma_{n,t_n}|>|\Sigma|-\epsilon\right\} \supset\left(\bigcap_{n=N}^\infty\left\{|\Sigma_{n,K_\epsilon}|>|\Sigma|-\epsilon\right\}\right)\bigcap\left(\bigcap_{n=N}^\infty\left\{s_n=m_0,t_n\geq K_\epsilon\right\}\right) .
\end{align*}
Therefore, to prove \eqref{eq:liminfty} it suffices to show
\begin{equation}\label{eq:part1}
\Pr\left(\bigcap_{n=N}^\infty\{ |\Sigma_{n,K_\epsilon}|>|\Sigma|-\epsilon\}\right)\to1 \text{ as }N\to\infty, \tag{A.24}
\end{equation}
and
\begin{equation}\label{eq:part2}
\Pr\left(\bigcap_{n=N}^\infty\{s_n=m_0,t_n\geq K_\epsilon\}\right)\to1 \text{ as }N\to\infty. \tag{A.25}
\end{equation}
By the continuity of measure, \eqref{eq:part2} is equivalent to Lemma
\ref{lemma:Indicator0} and thus holds true. Then it remains to prove \eqref{eq:part1}.

By Corollary \ref{corol:matconverge} and \ref{corol:cst_Sigmai}, with
probability 1,
$
|\Sigma_{n,K_\epsilon}| \to |\Sigma_{K_\epsilon}|\text{ as }n\to\infty,
$
which gives
\begin{equation}\label{eq:SigmaKepsiloncst}
\Pr\left(\bigcap_{n=N}^\infty\{ \mathrm{abs}(|\Sigma_{n,K_\epsilon}|-|\Sigma_{K_\epsilon}|)<\epsilon/2\}\right)\to1\text{ as }N\to\infty \tag{A.26}
\end{equation}
where $\mathrm{abs}(\cdot)$ denotes absolute value.

When $|\Sigma|-|\Sigma_{K_\epsilon}|<\epsilon/2$ as in \eqref{eq:Sigmatailbdd},
$
\mathrm{abs}(|\Sigma_{n,K_\epsilon}|-|\Sigma_{K_\epsilon}|)<\epsilon/2~ \text{ implies }|\Sigma_{n,K_\epsilon}|>|\Sigma|-\epsilon,
$
so
\begin{equation}\label{eq:sim}
\Pr \left(\bigcap_{n=N}^\infty\{\mathrm{abs}(|\Sigma_{n,K_\epsilon}|-|\Sigma_{K_\epsilon}|)<\epsilon/2\} \right)\leq \Pr\left(\bigcap_{n=N}^\infty\{|\Sigma_{n,K_\epsilon}|>|\Sigma|-\epsilon\}\right). \tag{A.27}
\end{equation}
Putting \eqref{eq:SigmaKepsiloncst} and \eqref{eq:sim} together, we obtain \eqref{eq:part1}.
\end{proof}

\subsection{Confidence region with the univariate approach}
\label{app:CR}
We briefly state here the current methods for constructing confidence regions with univariate estimators. Let $\sigma(i)^2$ denote the $(i,i)$th entry of $\Sigma$. We treat the problem as $p$ univariate cases, i.e., to estimate $\sigma(i)^2$ using univariate samples. Then we construct cube-shaped confidence regions.

Let $\mu_n(i)$ be the $i$th component of $\mu_n$, and $\sigma_n(i)^2$ be the estimator for $\sigma(i)^2$.
The uncorrected confidence region is given by
\[
C_n=\begin{Bmatrix}
\mu_n(1)\pm z_{1-\alpha/2}\sigma_n(1)/\sqrt{n}\\
\mu_n(2)\pm z_{1-\alpha/2}\sigma_n(2)/\sqrt{n}\\
\vdots\\
\mu_n(p)\pm z_{1-\alpha/2}\sigma_n(p)/\sqrt{n}
\end{Bmatrix}
\]
with a volume of
\[
\left(\frac{2z_{1-\alpha/2}}{\sqrt{n}}\right)^p\prod_{i=1}^n\sigma_n(i).
\]
The Bonferroni confidence region for $\mu$ is
\[
B_n=\begin{Bmatrix}
\mu_n(1)\pm z_{1-\alpha/2p}\sigma_n(1)/\sqrt{n}\\
\mu_n(2)\pm z_{1-\alpha/2p}\sigma_n(2)/\sqrt{n}\\
\vdots\\
\mu_n(p)\pm z_{1-\alpha/2p}\sigma_n(p)/\sqrt{n}
\end{Bmatrix}
\]
with a volume of
\[
\left(\frac{2z_{1-\alpha/2p}}{\sqrt{n}}\right)^p\prod_{i=1}^n\sigma_n(i).
\]

\end{document}